\documentclass[12pt,twoside]{article}

\usepackage{float}
\usepackage{graphicx}
\usepackage{epstopdf}
\usepackage{graphicx}
\usepackage{epic}
\usepackage{multirow}
\usepackage{pst-poly}  
\usepackage{pst-plot}  
\usepackage{pst-poly}  
\usepackage{tikz}
\usepackage{xcolor}
\usepackage{amsmath}
\usepackage{amssymb}
\usepackage{mathrsfs}
\usepackage{makecell}
\usepackage{threeparttable}
\usepackage[figuresright]{rotating}
\usetikzlibrary{arrows,shapes,chains}
\usepackage{bm}

\renewcommand{\paragraph}{\roman{paragraph}}
\usepackage[a4paper]{geometry}
\makeatletter
\renewcommand\title[1]{\gdef\@title{\reset@font\Large\bfseries #1}}
\renewcommand\section{\@startsection {section}{1}{\z@}%
	{-3.5ex \@plus -1ex \@minus -.2ex}%
	{2.3ex \@plus.2ex}%
	{\normalfont\large\bfseries}}
\renewcommand\subsection{\@startsection{subsection}{2}{\z@}%
	{-3ex\@plus -1ex \@minus -.2ex}%
	{1.5ex \@plus .2ex}%
	{\normalfont\normalsize\bfseries}}
\renewcommand\subsubsection{\@startsection{subsubsection}{3}{\z@}%
	{-2.5ex\@plus -1ex \@minus -.2ex}%
	{1.5ex \@plus .2ex}%
	{\normalfont\normalsize\bfseries}}

\def\@runningauthor{}\newcommand{\runningauthor}[1]{\def\runningauthor{#1}}
\def\@runningtitle{}\newcommand{\runningtitle}[1]{\def\runningtitle{#1}}

\renewcommand{\ps@plain}{%
	\renewcommand{\@evenhead}{\footnotesize\scshape \hfill\runningauthor\hfill}
	\renewcommand{\@oddhead}{\footnotesize\scshape \hfill\runningtitle\hfill}}

\g@addto@macro\bfseries{\boldmath}

\makeatother

\usepackage{amsthm,amsmath,amssymb}
\usepackage{cite}
\usepackage{graphicx}

\usepackage[colorlinks=true,citecolor=black,linkcolor=black,urlcolor=blue]{hyperref}

\theoremstyle{plain}
\newtheorem{theorem}{Theorem}[section]

\newtheorem{lemma}[theorem]{Lemma}

\theoremstyle{definition}
\newtheorem{definition}[theorem]{Definition}
\newtheorem{example}[theorem]{Example}

\theoremstyle{remark}

\runningauthor{}

\date{}

\begin{document}

	\title{On Euclidean, Hermitian and symplectic quasi-cyclic complementary dual codes}
	\author{Chaofeng Guan, Ruihu Li,
		Zhi Ma
	}
	
	\date{}
	\maketitle
	
	\begin{abstract}
		Linear complementary dual codes (LCD) intersect trivially with their dual. In this paper, we develop a new characterization for LCD codes, which allows us to judge the complementary duality of linear codes from the codeword level. Further, we determine the sufficient and  necessary conditions for one-generator quasi-cyclic codes to be LCD codes involving Euclidean, Hermitian, and symplectic inner products. Finally, we constructed many Euclidean, Hermitian and symmetric LCD codes with excellent parameters, some improving the results in the literature. Remarkably, we construct a symplectic LCD $[28,6]_2$ code with symplectic distance $10$, which corresponds to an trace Hermitian additive complementary dual $(14,3,10)_4$ code that outperforms the optimal quaternary Hermitian LCD $[14,3,9]_4$ code.
		
	\end{abstract}
	{\bf Keywords:} quasi-cyclic codes, Euclidean, Hermitian, symplectic, complementary dual codes.\\
	{\bf AMS Classification (MSC 2020)}: 94B05, 15B05, 12E10
	
	\section{Introduction}\label{sec1}
	Linear complementary dual codes (LCD) intersect their dual codes trivially.
	LCD codes have been used extensively in data storage, communication systems, consumer electronics, and cryptography \cite{Massey1992LinearCW,carlet2016complementary}. In \cite{Massey1992LinearCW}, Massey showed that LCD codes provide an optimal linear coding scheme for a two-user binary adder channel. 
	In  \cite{carlet2016complementary}, Carlet et al. studied the application of binary LCD codes in countering side channel and fault injection attacks. 
	Due to the critical application of LCD codes, much research has been conducted on LCD codes \cite{li2017lcd,dougherty2017combinatorics,galvez2018some,carlet2018euclidean,li2018hermitian,araya2019quaternary,LiThehulloftwo,li2023improved}. 
	
	Notably, in \cite{tang2018linear}, Carlet et al. proved that $\mathbb{F}_q$-codes, $q\ge 4$, are equivalent to Euclidean LCD codes, while $\mathbb{F}_{q^2}$-codes, $q\ge 3$, are equivalent to Hermitian LCD codes.
	 Therefore, most research on LCD codes is currently focused on small fields. 
	Bouyuklieva \cite{bouyuklieva_optimal_2021}, Harada \cite{harada_construction_2021}, Ishizuka et al. \cite{ishizuka_construction_2021,ishizuka2022construction}, Li et al. \cite{Li2022ImprovedLA,Li2023OnTC}, and wang et al. \cite{Wang2023NewCO} constructed much good LCD codes, established several LCD code tables with short lengths.
	In addition, Shi et al. \cite{ShiAdditive_complementary2023,Shi2022AdditiveCD}  introduced additive complementary dual codes (ACD) for security applications that still make sense. With \cite{calderbank1998quantum}, binary symplectic inner product and quaternary trace Hermitian inner product are equivalent, so a $\mathbb{F}_2^{2n}$-symplectic LCD code is equivalent to a $\mathbb{F}_{4}^{n}$-trace Herimtian LCD code.
	Therefore, the construction of the symplectic LCD code is also of significant importance. 
	Xu et al. \cite{xu_constructions_2021} constructed a class of symplectic LCD MDS codes.
	In \cite{Huang2022}, Huang et al. construct some good low-dimensional quasi-cyclic symplectic LCD codes over $\mathbb{F}_{2^{r}}$.

	In \cite{Yang1994TheCF}, Yang and Massey provided a sufficient and necessary condition for cyclic codes to be Euclidean LCD codes.
	Esmaeili et al. \cite{Esmaeili2009OnCQ} studied a sufficient condition for quasi-cyclic codes to be Euclidean LCD codes and gave a method for constructing quasi-cyclic Euclidean LCD codes.
	In \cite{guneri2016quasi}, Güneri et al. studied quasi-cyclic codes with Euclidean and Hermitian complementary duals employing their concatenation structure and obtained the sufficient and necessary condition for a class of one-generator quasi-cyclic codes with index two to be LCD codes.
	Later, Alahmadi et al. \cite{alahmadi2020complementary} propose the sufficient and necessary condition for a class of one-generator multinegacirculant codes (a subclass of quasi-twisted codes with twisting constant $-1$) with index $t$ to be Euclidean LCD.
	However, many important issues remain regarding developing LCD codes from quasi-cyclic codes. One of the most critical issues is determining the sufficient and necessary conditions for quasi-cyclic codes to be LCD codes so that we can construct quasi-cyclic LCD codes more efficiently.

	The main goal of this paper is to investigate quasi-cyclic Euclidean, Hermitian, and symplectic LCD codes. We ascertain the sufficient and necessary conditions for one-generator quasi-cyclic codes to be  Euclidean, Hermitian, and symplectic LCD codes. More precisely, we answer the following two questions:
	
	\textbf{1. What polynomials can be applied to construct quasi-cyclic LCD codes?}
	
	\textbf{2. How to use polynomials to construct quasi-cyclic LCD codes?}
	
	Firstly, we give a new characterization of LCD codes that allows the treatment of LCD codes at the codeword level. 
	Then, by decomposing the codeword space of quasi-cyclic code, we obtain the sufficient and necessary conditions for them to intersect their dual trivially under Euclidean, Hermitian, and symplectic inner products. 
	Finally, we present a practical method for constructing LCD codes using quasi-cyclic codes and construct many good quasi-cyclic Euclidean, Hermitian, and symplectic LCD codes.
	
	The paper is organized as follows. Section \ref{sec2} gives preliminaries and background on quasi-cyclic codes, Euclidean, Hermitian, and symplectic LCD codes. In Section \ref{sec3} and \ref{sec4}, we redescribe LCD codes in terms of codewords and identify the sufficient and necessary conditions for the quasi-cyclic codes to be LCD codes under Euclidean, Hermitian, and symplectic inner products, respectively. 
	We also construct many good quasi-cyclic Euclidean, Hermitian, and symplectic LCD codes.
	Finally, we give concluding remarks in Section \ref{sec6}. 
	All calculations in this paper are done with the algebraic computer system Magma \cite{bosma1997magma}.

	\section{Preliminaries}\label{sec2}
	In this section, we introduce some basic concepts of quasi-cyclic codes, Euclidean, Hermitian, and symplectic LCD codes. For more details, we refer the reader to the standard handbooks \cite{huffman2010fundamentals,huffman2021concise}.
	\subsection{Basics of linear codes}
	Throughout this paper, $p$ is a prime, and $\mathbb{F}_q$ is the finite field of order $q$, where $q=p^r$ for some positive integer $r$.
	A $[\ell n, k]_q$ linear code $\mathscr{C}$ over $\mathbb{F}_{q}$ is a linear subspace of $\mathbb{F}_{q}^{\ell n}$ of dimension $k$.
	Let $\bm{u}=(u_0,u_1,\dots,u_{\ell n -1}) \in \mathscr{C}$, then Hamming weight of $\bm{u}$ is $\mathbf{w}_{H}(\bm{u})=\#\{i \mid u_{i}\neq0 , 0 \leq i \leq \ell n-1\}$. 
	If minimum Hamming distance of $\mathscr{C}$ is $d_H=\min\{\mathbf{w}_{H}(\bm{u})\mid \bm{u} \in \mathscr{C}\setminus \{\mathbf{0}\}\}$, then $\mathscr{C}$ can be written as $[\ell n,k,d_H]_q$.
	If $\ell$ is even, let $N=\ell n/2$, then symplectic weight of $\bm{u}$ is 
	$\mathbf{w}_{s}(\bm{u})=\#\{i \mid (u_{i}, u_{N+i}) \neq(0,0) , 0 \leq i \leq N-1\}$.
	Analogously, if minimum symplectic weight of $\mathscr{C}$ is $d_{s}(\mathscr{C})=\min \{\mathbf{w}_{s}(\bm{u}) \mid \bm{u} \in \mathscr{C}\setminus \{\mathbf{0}\}\}$, then we denote $\mathscr{C}$ as $[\ell n,k,d_s]_q^s$. 
	
	The Euclidean inner product of  $\bm{x}=(x_{0},  \ldots, x_{\ell n-1})$, 
	$\bm{y}=(y_{0},\ldots, y_{\ell n-1}) \in \mathbb{F}_q^{\ell n}$ is defined as:
	
	\begin{equation}
		\langle\bm{x}, \bm{y}\rangle_{e}=\sum_{i=0}^{\ell n-1} x_{i} y_{i}.
	\end{equation} 
	
	Similarly, the Hermitian inner product of  $\bm{x},\bm{y} \in \mathbb{F}_{q^2}^{\ell n}$ is defined as:
	\begin{equation}
		\langle\bm{x}, \bm{y}\rangle_{h}=\sum_{i=0}^{\ell n-1} x_{i} y_{i}^{q}.
	\end{equation} 
	
	If $\ell$ is even, then symplectic inner product of $\bm{x},\bm{y} \in \mathbb{F}_{q}^{\ell n}$ is:
	
	\begin{equation}
		\langle \bm{x},\bm{y}\rangle_{s}=\sum_{i=0}^{N-1}\left(x_{i} y_{N+i}-x_{N+i} y_{i}\right).
	\end{equation}
	
	%
	%

	The Euclidean dual of $\mathscr{C}$ is $\mathscr{C}^{\perp_{e}}=\{\bm{v} \in \mathbb{F}_{q}^{\ell n} \mid\langle\bm{u}, \bm{v}\rangle_{e}=0, \forall \bm{u} \in \mathscr{C}\}$; if $\ell n$ is even, then the symplectic dual of $\mathscr{C}$ is $\mathscr{C}^{\perp_{s}}= \{\bm{v} \in \mathbb{F}_{q}^{\ell n} \mid\langle\bm{u}, \bm{v}\rangle_{s}=0, \forall \bm{u} \in \mathscr{C}\}$; 
	if $q=p^2$, then the Hermitian dual of $\mathscr{C}$ is $\mathscr{C}^{\perp_{h}}=\{\bm{v} \in \mathbb{F}_{q^2}^{\ell n} \mid\langle\bm{u}, \bm{v}\rangle_{h}=0, \forall \bm{u} \in \mathscr{C}\}$.
	$\mathscr{C}$ is an LCD code if and only if $\mathscr{C} \cap \mathscr{C}^{\perp_*}=\{\mathbf{0}\}$, where ``$\perp_*$" represents one of Euclidean, Hermitian and symplectic dual. 
	\subsection{Basics of quasi-cyclic codes}
	Cyclic codes are a class of linear codes closed under the shift operator $\tau _{1}$. 
For $\bm{x}=(x_{0}, x_{1}, \ldots, x_{n-1}) \in \mathbb{F}_{q}^{n}$, we denote
$\tau _{1}(\bm{x})=\left(x_{n-1}, x_{0}, \ldots, x_{n-2}\right).$
If $\mathscr{C}=\tau _{1}(\mathscr{C})$, then $\mathscr{C}$ is called a cyclic code. 
Let $\mathbb{R}=\mathbb{F}_{q}[x] /\left\langle x^{n}-1\right\rangle$, and define a mapping $\varphi_{1}$ as follows,
\begin{equation}
	\begin{aligned}
		\varphi_{1}: \mathbb{F}_{q}^{n} & \rightarrow \mathbb{R} \\
		\left(c_{0}, c_{1}, \ldots, c_{n-1}\right) & \mapsto c_{0}+c_{1} x+\cdots+c_{n-1} x^{n-1}
	\end{aligned}
\end{equation}

Clearly, $\varphi_{1}$ is an isomorphism of $\mathbb{F}_{q}$-modules and a cyclic code $\mathscr{C}$ of length $n$ is an ideal of the quotient ring $\mathbb{R}$.
Furthermore, a cyclic code $\mathscr{C}$ can be generated by a monic divisor $g(x)$ of $x^n -1$. The polynomial $g(x)$ is called the generator polynomial of $\mathscr{C}$, and the dimension of  $\mathscr{C}$ is $n-deg(g(x))$. 
Let $h(x)=x^n-1/g(x)$ and $\tilde{h}(x)=x^{\deg(h(x))} h(x^{-1})$, then Euclidean dual code of $\mathscr{C}$ is cyclic code with generator polynomial $g^{\perp_e}(x)=\tilde{h}(x)$.
Let $g^{q}(x)=g^q_{0}+g^q_{1} x+g^q_{2} x+\cdots+ g^q_{n-1} x^{n-1}$.
If $\mathscr{C}$ is a cyclic codes over $\mathbb{F}_{q^2}$, then Hermitian dual code of $\mathscr{C}$ is cyclic code generated by $g^{\perp_h}(x)=\tilde{h}^q(x)$.

Let $\bm{x}=\left(x_{0}, x_{1}, \ldots, x_{\ell n-1}\right) \in \mathbb{F}_{q}^{\ell n}$, and
$\tau _{2}(\bm{x})=(x_{n-1},x_0, \ldots, x_{n-2}, x_{2n-1}, x_{n}, \ldots, x_{2n-2},$ $\ldots,  x_{n\ell-1}, x_{(n-1)\ell}, \ldots, x_{n\ell-2}).$
If $\mathscr{C}=\tau _{2}(\mathscr{C})$, A linear space $\mathscr{C} \subset \mathbb{F}_{q}^{\ell n}$  said to be a quasi-cyclic code of index $\ell$. Define an $\mathbb{F}_q$ -module isomorphism $\varphi_{2}$ from $\mathbb{F}_{q}^{\ell n}$ to $\mathbb{R}^{\ell}$,

$$\begin{aligned}
	&	\varphi_{2}:   \mathbb{F}_{q}^{\ell n}  \rightarrow \mathbb{R}^{\ell}=\mathbb{R}\oplus \mathbb{R} \oplus \dots \oplus \mathbb{R}\\
	&\left(c_{0}, \ldots, c_{n-1}, c_{n}, \ldots, c_{2n-1}, \ldots, c_{\ell (n-1)}, \ldots, c_{\ell n-1}\right) \\
	&\mapsto\left(c_{0}(x), c_{1}(x), \ldots, c_{\ell-1}(x)\right),
\end{aligned}$$
	where  $c_{i}(x)=\sum_{t=0}^{n-1} c_{i, t} x^{t}, i=0,1, \ldots, \ell-1$.
	Algebraically, a quasi-cyclic code $\mathscr{C}$ is a submodule of $\mathbb{R}^{\ell}$.
	\section{New characterization of complementary dual codes}\label{sec3}
	This section gives a new characterization of LCD codes in terms of codewords, laying the foundation for further proof. First, we make a convention for representing inner products, where ``$l$" denotes one of Euclidean or Hermitian inner products, and ``$*$" denotes one of Euclidean, Hermitian, and symplectic products.

	\begin{lemma}\label{LCD_def}
		Let $\mathscr{C}$ be a linear code over $\mathbb{F}_{q}$, then $\mathscr{C}$ is an LCD code under the inner product ``$*$" if and only if $\forall c_1\in \mathscr{C}\setminus \{\mathbf{0}\}$, $\exists c_2 \in \mathscr{C}, \langle c_1,  c_2\rangle_{*}\ne0$ holds.
	\end{lemma}
	\begin{proof}
		It is obvious that $\mathscr{C}\cap\mathscr{C}^{\perp_*}=\{\mathbf{0}\}$ is equivalent with $\forall c_1\in \mathscr{C}\setminus \{\mathbf{0}\}$, $c_1\notin \mathscr{C}^{\perp_*}$. Moreover, $c_1\notin \mathscr{C}^{\perp_*}$ is equivalent with $\exists c_2 \in \mathscr{C}, \langle c_1,  c_2 \rangle_{*}\ne0$, so $\mathscr{C}$ is LCD equivalent with $\forall c_1 \in \mathscr{C}\setminus \{\mathbf{0}\}$, $\exists c_2 \in \mathscr{C}, \langle c_1,  c_2 \rangle_{*}\ne0$.
	\end{proof}
	For ease of presentation, we give the following definition.
	\begin{definition}
		Let $\mathscr{C}_1$ and $\mathscr{C}_2$ be linear codes over $\mathbb{F}_{q}$, if the following conditions hold:
		
		\begin{equation}
			\left\{
			\begin{array}{c} 
				\forall c_1\in \mathscr{C}_1\setminus \{\mathbf{0}\}, \exists c_2 \in \mathscr{C}_2, \langle c_1,  c_2\rangle_{*}\ne0, \\ 
				\forall c_2\in \mathscr{C}_2\setminus \{\mathbf{0}\}, \exists c_1 \in \mathscr{C}_1, \langle c_1,  c_2\rangle_{*}\ne0,
			\end{array}
			\right.
		\end{equation}	
		then we call $\mathscr{C}_1$ and $\mathscr{C}_2$ \textbf{completely non-orthogonal} to each other under inner product ``$*$".
	\end{definition}
	
	\begin{lemma}
		Let $\mathscr{C}_1$ and $\mathscr{C}_2$ be linear codes over $\mathbb{F}_{q}$, then $\mathscr{C}_1$ and $\mathscr{C}_2$ completely non-orthogonal to each other holds if and only if $\mathscr{C}_1 \cap \mathscr{C}_2^{\perp_*}=\{\mathbf{0}\}$ and $\mathscr{C}_1^{\perp_*} \cap \mathscr{C}_2=\{\mathbf{0}\}$.
	\end{lemma}
	\begin{proof}
		This lemma holds from the definition of dual codes.
	\end{proof}
	\begin{lemma}\label{cyclic_LCD_def}
		Let $\mathscr{C}_1$ and $\mathscr{C}_2$ be two cyclic codes , and separately generated by $g_1(x)$ and $g_2(x)$, where $g_1(x)\mid x^n-1$ and $g_2(x)\mid x^n-1$, then $\mathscr{C}_1$ and $\mathscr{C}_2$ completely Euclidean non-orthogonal to each other is equivalent with $g_1(x)=\tilde{g} _1(x)=g_2(x)$;
		$\mathscr{C}_1$ and $\mathscr{C}_2$ completely Hermitian non-orthogonal to each other is equivalent with $g_1(x)=\tilde{g}^q _1(x)=g_2(x)$.
	\end{lemma}
	\begin{proof}
		With \cite{huffman2010fundamentals}, $\mathscr{C}_1 \cap \mathscr{C}_2^{\perp_l}$ and $\mathscr{C}_1^{\perp_l} \cap \mathscr{C}_2$ are both cyclic codes generated by $lcm(g_1(x),g_2^{\perp_l}(x))$ and $lcm(g_1^{\perp_l}(x),g_2(x))$, respectively.  Further, $\mathscr{C}_1 \cap \mathscr{C}_2^{\perp_l}=\{\mathbf{0}\}$ and $\mathscr{C}_1^{\perp_l} \cap \mathscr{C}_2=\{\mathbf{0}\}$ yield $lcm(g_1(x),g_2^{\perp_l}(x))\equiv 0 \pmod{x^n-1} $ and $lcm(g_1^{\perp_l}(x),g_2(x))\equiv 0 \pmod{x^n-1} $. Thus, there are two cases, the first is $\tilde{g}_2(x)\mid g_1(x) $ and $\tilde{g}_1(x)\mid g_2(x) \Longleftrightarrow  g_1(x)=\tilde{g} _1(x)=g_2(x)$. The second is $\tilde{g}^q_2(x)\mid g_1(x) $ and $\tilde{g}^q_1(x)\mid g_2(x) \Longleftrightarrow  g_1(x)=\tilde{g}^q _1(x)=g_2(x)$. Hence, we complete the proof.
	\end{proof}

	\section{Quasi-cyclic complementary dual codes}\label{sec4}
This section determines the sufficient and necessary conditions for one-generator quasi-cyclic codes to be LCD codes under Euclidean, Hermitian, and symplectic inner products, starting from Lemma \ref{LCD_def}.

Some symbols used are described below for ease of expression.
Let $g(x)=g_{0}+g_{1} x+g_{2} x+\cdots+ g_{n-1} x^{n-1} \in \mathbb{R}$, $[g(x)]$ denote vector defined by coefficients of $g(x)$ in $\mathbb{F}_{q}^{n}$, i.e. $[g(x)]=[{{g}_{0}},{{g}_{1}},{{g}_{2}},\cdots ,{{g}_{n-1}}]$, and $\bar{g}(x)=x^ng(x^{-1})$.

In order to determine the Euclidean and Hermitian inner products between different polynomials in coefficient vector form, the following two lemmas are crucial.

\begin{lemma} \label{commutative_law}(\cite{galindo2018quasi})
	Let $f(x)$, $g(x)$ and $h(x)$ be polynomials in $\mathbb{R}$. Then the following equation holds for the Euclidean inner product among them: 
	\begin{equation}
		\langle[f(x) g(x)],[h(x)]\rangle_{e}=\langle[g(x)],[\bar{f}(x) h(x)]\rangle_{e}.
	\end{equation}
\end{lemma}
\begin{lemma}(\cite{lv2020quantum})\label{exchange_law}
	Let  $f(x)$, $g(x)$  and  $h(x)$  be monic polynomials in  $\mathbb{R}$. Then the following equality of Hermitian inner product of vectors in  $\mathbb{F}_{q^{2}}^{n}$  holds:
	\begin{equation}
		\langle[f(x) g(x)],[h(x)]\rangle_{h}=\langle[g(x)],[\bar{f}^{q}(x) h(x)]\rangle_{h}.
	\end{equation}
\end{lemma}


	\begin{definition}\label{one_quasi-cyclic_def}
		Let  $g(x)$ and $f_{j}(x)$ be monic polynomials in $\mathbb{R}$, and $g(x)\mid(x^{n}-1)$, $0\le j\le \ell -1$. If $\mathscr{C}$ is a quasi-cyclic code generated by $([g(x){{f}_{0}}(x)]$, $[g(x){{f}_{1}}(x)]$, $\cdots$, $[g(x){{f}_{\ell -1}}(x)])$, then $\mathscr{C}$ is called one-generator quasi-cyclic code with index $\ell$.
		A genrartor matrix $G$ of $\mathscr{C}$  has the following form: 
		\begin{equation}
			G=\left(G_{0}, G_{1}, \cdots, G_{\ell-1}\right),
		\end{equation}
		where $G_ {j} $ are $n\times n$  circulant matrices generated by $ [g (x) f_{j} (x)] $, respectively.
	\end{definition}

	\begin{theorem}\label{onegenearator_sufficient}
		If $\mathscr{C}$ is a one-generator quasi-cyclic code in Definition \ref{one_quasi-cyclic_def}, then the sufficient and necessary conditions for $\mathscr{C}$ to be \textbf{Euclidean LCD} code are
		\begin{equation}
			\left\{
			\begin{array}{c} 
				g(x)=\tilde{g}(x), \\ 
				gcd(\sum\limits_{i=0}^{\ell-1}f_i(x)\bar{f}_i(x),\frac{x^n-1}{g(x)})=1.
			\end{array}
			\right.
		\end{equation}
	\end{theorem}
	\begin{proof}
		Suppose $a(x)$,$b(x)$ are any polynomials in $\mathbb{R}$, then any two codewords in $\mathscr{C}$ can be represented as 
		$\bm{c_1}=([a(x)g(x){{f}_{0}}(x)], [a(x)g(x){{f}_{1}}(x)],\cdots,[a(x)g(x){{f}_{\ell -1}}(x)])$ and 
		$\bm{c_2}=([b(x)g(x){{f}_{0}}(x)],[b(x)g(x){{f}_{1}}(x)],\cdots,[b(x)g(x){{f}_{\ell -1}}(x)])$, respectively.
		The Euclidean inner product of $c_1$ and $c_2$ can be expressed as:
		
		$\begin{array}{rl} 
			\langle \bm{c_1},  \bm{c_2}\rangle_{e}=& \sum\limits_{i=0}^{\ell-1} \langle [a(x)g(x){{f}_{i}}(x)],  [b(x)g(x){{f}_{i}}(x)]\rangle_{e}\\ 
			=& \sum\limits_{i=0}^{\ell-1} \langle [a(x)g(x){{f}_{i}}(x)\bar{f}_i(x)],  [b(x)g(x)]\rangle_{e}\\ 
			=&  \langle [a(x)g(x)\sum\limits_{i=0}^{\ell-1}{{f}_{i}}(x)\bar{f}_i(x)],  [b(x)g(x)]\rangle_{e}.\\ 
		\end{array}$
		
		From Lemma \ref{LCD_def} and \ref{cyclic_LCD_def}, it is clear that the sufficient and necessary conditions for $\mathscr{C}$ to be Euclidean LCD code is that $g(x)=\tilde{g}(x)$ and $gcd(\sum\limits_{i=0}^{\ell-1}f_i(x)\bar{f}_i(x),\frac{x^n-1}{g(x)})=1$, so this theorem is proved.
	\end{proof}


		In 	\cite{guneri2016quasi,alahmadi2020complementary}, there are also a sufficient and necessary conditions for a class of one-generator negacirculant codes to be Euclidean LCD, as follows.
		\begin{lemma}(\cite{guneri2016quasi,alahmadi2020complementary})\label{nega-quasi}
			Suppose that $\mathscr{C}_n$ is a one-generator negacirculant code with index $t$ generated by $\left(1, a_{1}(x), \ldots, a_{t-1}(x)\right)$ and $a_i(x) \in \mathbb{F}_{q}[x] /\left\langle x^{n}+1\right\rangle$. For $gcd(n,q)=1$, the sufficient and necessary condition for $\mathscr{C}_n$ to be Euclidean LCD is 
			\begin{equation}
				gcd\left(1+\sum_{y=1}^{t-1} a_{y}(x) a_{y}\left(-x^{n-1}\right), x^{n}+1\right)=1.
			\end{equation}
		\end{lemma}
		
		It is notable that the results for \cite{guneri2016quasi,alahmadi2020complementary} simplify the structure of quasi-cyclic codes, so that Lemma \ref{nega-quasi} can only construct Euclidean LCD codes with a rate of $1/t$, and length of $tn$, where $gcd(n,q)=1$. In contrast, it is possible to construct LCD codes with very flexible length and dimensionality by varying $g(x)$ according to Lemma \ref{onegenearator_sufficient}. For a clearer comparison, we give the following example.
		\begin{example}\label{Compar_ternary}
			Let $q=3$ and $n=13$, only one Euclidean LCD $[26,13,8]_3$ code can be constructed in \cite{alahmadi2020complementary}. Here, with Lemma \ref{onegenearator_sufficient}, choosing different ideals, we get nine good Euclidean LCD codes with parameters: 
			$[26,4,15]_3$, $[26,7,13]_3$, $[26,8,11]_3$, $[26,9,10]_3$, $\bm{[26,12,9]_3}$, $\bm{[26,13,8]_3}$, $\bm{[26,14,7]_3}$, $[26,15,6]_3$, $\bm{[26,20,4]_3}$. The bolded codes are the optimal or best-known according to \cite{Grassltable}.  In addition, we give the constructions for these codes in Appendix A.
			
			In order to show the effectiveness of our method, we also constructed 17 new binary Euclidean LCD codes, whose parameters are given in Table \ref{New_B_LCD}. To save space, the detailed construction methods are given in Appendix B.
			
		\end{example}
	
		\begin{table*}[!t]
		\centering
		\begin{threeparttable}
			\caption{New Binary Quasi-Cyclic Euclidean LCD Codes}\label{New_B_LCD}
			\centering
			\begin{tabular}{ccc}
				\hline
				No. &     Our LCD Codes      & Best LCD Codes in \cite{Li2022ImprovedLA, Wang2023NewCO} \\ \hline
				1  &     $[39,13,12]_2$     &                      $[39,13,11]_2$                      \\
				2  &     $[44,11,16]_2$     &                      $[44,11,15]_2$                      \\
				3  &     $[44,12,15]_2$     &                      $[44,12,14]_2$                      \\
				4  &     $[44,22,9]_2$      &                      $[44,22,8]_2$                       \\
				5  &     $[45,12,16]_2$     &                      $[45,12,15]_2$                      \\
				6  &     $[45,13,15]_2$     &                      $[45,13,14]_2$                      \\
				7  &     $[45,22,10]_2$     &                      $[45,22,9]_2$                       \\
				8  &     $[45,23,9]_2$      &                      $[45,23,8]_2$                       \\
				9  &     $[46,23,10]_2$     &                      $[46,23,9]_2$                       \\
				10  &     $[46,24,9]_2$      &                      $[46,24,8]_2$                       \\
				11  &     $[49,15,15]_2$     &                      $[49,15,14]_2$                      \\
				12  &     $[50,8,21]_2$      &                      $[50,8,20]_2$                       \\
				13  &     $[50,15,16]_2$     &                      $[50,15,14]_2$                      \\
				14  &     $[50,16,15]_2$     &                      $[50,16,14]_2$                      \\
				15  & $[51,8,22]_2^{\star}$  &                            -                             \\
				16  & $[51,16,16]_2^{\star}$ &                            -                             \\
				17  & $[51,17,15]_2^{\star}$ &                            -                             \\ \hline
			\end{tabular}
			\begin{tablenotes}    
				\footnotesize               
				\item[$\star$] Since new binary LCD codes can be derived from these codes, they are also new.        
			\end{tablenotes}            
		\end{threeparttable}       
	\end{table*}
		
	Since the Hermitian inner product and the Euclidean inner product have a similar form, an analogous approach yields the sufficient and necessary conditions for 1-generator quasi-cyclic code to be a Hermitian LCD code. Therefore, we give the following theorem without proof.
	
	\begin{theorem}
		If $\mathscr{C}$ is a one-generator quasi-cyclic code in Definition \ref{one_quasi-cyclic_def}, then the sufficient and necessary conditions for $\mathscr{C}$ to be \textbf{Hermitian LCD} code are 
		\begin{equation}
			\left\{
			\begin{array}{c} 
				g(x)=\tilde{g}^q(x), \\ 
				gcd(\sum\limits_{i=0}^{\ell-1}f_i(x)\bar{f}^q_i(x),\frac{x^n-1}{g(x)})=1.
			\end{array}
			\right.
		\end{equation}		
	\end{theorem}
	
		Let $\{0,1,w,w^2\}$ denote elements in $\mathbb{F}_4$, where $w$ satisfies $w^2 + w + 1 = 0$. We also construct some good quaternary quasi-cyclic Hermitian LCD codes, and six of them are new. Details are listed in Table \ref{Good_Q_LCD}.

	\begin{table*}[!t]
		\centering
		\begin{threeparttable}
			\caption{Good Quaternary Quasi-Cyclic Hermitian LCD Codes}\label{Good_Q_LCD}
			\centering
			\begin{tabular}{ccc}
				\hline
				No. &   Our LCD Codes    &                              Constructions                               \\ \hline
				 1  &    $[10,4,6]_4$    &                      $x + 1,wx^3 + x^2 + wx + w^2$                       \\
				 2  &    $[14,6,7]_4$    &                  $x + 1,w^2x^5 + w^2x^4 + wx^2 + x + 1$                  \\
				 3  &    $[14,7,6]_4$    &                       $1,x^4 + x^3 + wx^2 + x + 1$                       \\
				 4  &   $[18,5,10]_4$    &          $x^4 + x^3 + w^2x^2 + wx + w,wx^4 + wx^2 + wx$           \\
				 5  &    $[20,8,9]_4$    &               $x^2 + 1,x^7 + w^2x^5 + x^4 + wx^3 + wx^2;$                \\
				 6  & $\bm{[20,8,9]_4}$  &                           Shorten $[21,9,9]_4$                           \\
				 7  & $\bm{[21,9,9]_4}$  &                          Shorten $[22,10,9]_4$                           \\
				 8  & $\bm{[22,9,9]_4}$  &                           Extend $[21,9,9]_4$                            \\
				 9  & $\bm{[22,10,9]_4}$ &      $x + 1,wx^9 + w^2x^7 + w^2x^6 + wx^5 + w^2x^4 + wx^3 + x + w;$      \\
				10  & $\bm{[23,10,9]_4}$ &                           Extend $[22,10,9]_4$                           \\
				11  & $\bm{[22,11,8]_4}$ & $1,wx^{10} + x^9 + w^2x^8 + w^2x^6 + x^5 + x^4 + x^3 + x^2 + w^2x + w^2$ \\ \hline
			\end{tabular}
			\begin{tablenotes}    
				\footnotesize               
				\item[]Note: All the quasi-cyclic LCD codes in this table with generator $([g(x)],([g(x)f(x)]))$. Refer to Ref. \cite{ishizuka2022constructionarxiv}, bolded codes are new quaternary Hermitian LCD codes; others reach the best-known lower bound on minimum distances in \cite{ishizuka2022constructionarxiv}.        
			\end{tablenotes}            
		\end{threeparttable}       
	\end{table*}

	\begin{theorem}\label{one_quasi-cyclic}
		If $\mathscr{C}$ is a one-generator quasi-cyclic code in Definition \ref{one_quasi-cyclic_def},  and $\ell$ is even. Let $m=\ell/2$, then $\mathscr{C}$ is a \textbf{symplectic LCD}  code if and only if the following equations hold.
		
		\begin{equation}
			\left\{
			\begin{array}{c} 
				g(x)=\tilde{g}(x), \\ 
				gcd(\sum\limits_{j=0}^{m-1}({{f}_{j}}(x){{\bar{f}}_{m+j}}(x)-{{f}_{m+j}}(x){{\bar{f}}_{j}}(x)),\frac{x^n-1}{g(x)})=1.
			\end{array}
			\right.
		\end{equation}	
	\end{theorem}
	
	\begin{proof}
		Suppose $a(x)$,$b(x)$ are any polynomials in $\mathbb{R}$, then any two codewords in $\mathscr{C}$ can be represented as 
		$\bm{c_1}=([a(x)g(x){{f}_{0}}(x)], [a(x)g(x){{f}_{1}}(x)],\cdots,[a(x)g(x){{f}_{\ell -1}}(x)])$ and 
		$\bm{c_2}=([b(x)g(x){{f}_{0}}(x)],[b(x)g(x){{f}_{1}}(x)],\cdots,[b(x)g(x){{f}_{\ell -1}}(x)])$, respectively.
		The symplectic inner product of $c_1$ and $c_2$ can be expressed as:
		
		$\begin{array}{rl} 
			\langle \bm{c_1},  \bm{c_2}\rangle_{s}=&  \bm{c_1} \cdot\left(\begin{array}{cc}
				0 & I_{m n} \\
				-I_{m n} & 0
			\end{array}\right) \cdot \bm{c_2}^{T}\\ 
			=&  \sum\limits_{j = 0}^{m - 1} {{{\langle {[ {a(x)g(x){f_j}(x)} ],[ {b(x)g(x){f_{m + j}}(x)} ]} \rangle }_e}} \\ 
			&- \sum\limits_{j = 0}^{m - 1} {{{\langle {[ {a(x)g(x){f_{m + j}}(x)} ],[ {b(x)g(x){f_j}(x)} ]} \rangle }_e}}\\
		 =&\sum\limits_{j = 0}^{m - 1} {{{\langle {[ {a(x)g(x){f_j}(x){{\bar f}_{m + j}}(x)} ],[ {b(x)g(x)} ]} \rangle }_e}} \\
		 & - \sum\limits_{j = 0}^{m - 1} {{{\langle {[ {a(x)g(x){f_{m + j}}(x){{\bar f}_j}(x)} ],[ {b(x)g(x)} ]} \rangle }_e}}\\
		 =&  \langle [ a(x) g (x)\sum\limits_{j = 0}^{m - 1}( {f_j}(x){{\bar f}_{m + j}}(x) - {f_{m + j}}(x){{\bar f}_j}(x) ) ],[b(x)g(x)] \rangle_e.  
		\end{array}$
	

		From Lemma \ref{LCD_def} and \ref{cyclic_LCD_def}, it is clear that the sufficient and necessary conditions for $\mathscr{C}$ to be symplectic LCD code is that $g(x)=\tilde{g}(x)$ and $gcd(\sum\limits_{j=0}^{m-1}({{f}_{j}}(x){{\bar{f}}_{m+j}}(x)-{{f}_{m+j}}(x){{\bar{f}}_{j}}(x)),\frac{x^n-1}{g(x)})=1$, so this theorem is proved.
	\end{proof}

	Theorem \ref{one_quasi-cyclic} determines the sufficient and sufficient conditions for one-generator quasi-cyclic code of even index to be symplectic LCD.
	In \cite{Huang2022}, Huang et al. also determined the condition for quasi-cyclic code $\mathscr{C}$ generated by $([g(x)],[g(x)f(x)])$ to be symplectic LCD code by deriving the relationship between $\mathscr{C}$ and its symplectic dual code $\mathscr{C}^{\perp_s}$. 
	
	\begin{lemma}(\cite{Huang2022}, Theorem 3.1) \label{Huang_QC}
		Let  $\mathscr{C}(f, g)$  be a  quasi-cyclic code over  $\mathbb{F}_{q}$  of length  $2 n$  generated by  $(g(x), f(x) g(x))$, where  $gcd\left(g(x), g^{\perp}(x)\right)=1$, $h(x) \mid g^{\perp}(x), \operatorname{gcd}(h(x), \bar{f}(x)-f(x))=1 $. Then  $\mathscr{C}(f, g)$  is a symplectic  LCD  code.
	\end{lemma}
	However, Lemma \ref{Huang_QC} is only a sufficient condition. By Theorem \ref{one_quasi-cyclic}, the sufficient and necessary conditions for $\mathscr{C}(f, g)$ to by symplectic LCD are as the following equation.
	\begin{equation}
		\left\{
		\begin{array}{c} 
			g(x)=\tilde{g}(x), \\ 
			gcd(f(x)-\bar{f}(x),\frac{x^n-1}{g(x)})=1.
		\end{array}
		\right.
	\end{equation}		

	Since binary symplectic inner product and quaternary trace Hermitian inner product are equivalent, one crucial motivation for symplectic LCD codes is to construct trace Hermitian ACD codes that have better performance than quaternary Hermitian LCD codes.
	The following two examples demonstrate how trace Hermitian ACD codes can be constructed to outperform quaternary LCD codes. For more, we refer the readers to see \cite{Guansomeadditive}.
		\begin{example}
		Set $q=2$, $n=7$, $\ell=4$. Let $g(x)=x + 1$, $f_1(x)=x^5 + x^3 + x^2$, $f_2(x)=x^4 + x^3$, $f_3(x)=x^6 + x^5 + x^4 + x$. 
		One can easy to check that $g(x)=\tilde{g}(x)$, $gcd(\sum\limits_{j=0}^{1}({{f}_{j}}(x){{\bar{f}}_{1+j}}(x)-{{f}_{1+j}}(x){{\bar{f}}_{j}}(x)),\frac{x^n-1}{g(x)})=1$, so $([g(x)f_0(x)],[g(x)f_1(x)],[g(x)f_2(x)]$, $[g(x)f_3(x)])$ can generate a 1-generator quasi-cyclic symplectic LCD code. Then, using Magma \cite{bosma1997magma} we can calculate this code have parameters $[28,6,10]_2^s$, whose symplectic weight distribution is  
		$\mathbf{w}_s(z)=1+35z^{10}+21z^{11}+7z^{13}$.
		Therefore, a trace Hermitian ACD code with parameters $(14,3,10)_4$ exists. It should be noted that the \textbf{optimal Hermitian LCD code} in \cite{ArayaQuaternaryLCD} with length 14 and dimension 3 has parameters $[14,3,9]_4$, so our symplectic construction has better performance.
	\end{example}

	\begin{example}
		Set $q=2$, $n=21$, , $\ell=2$. Let $g(x)=x^3 + 1$, $f_0(x)=x^{18} + x^{16} + x^{15} + x^{14} + x^{13} + x^{12} + x^8 + x^7 + x^3 + 1$, $f_1(x)=x^{20} + x^{19} + x^{18} + x^{15} + x^{14} + x^9 + x^7 + x^5 + x^3 + x^2 + 1$. 
		One can easy to check that $g(x)=\tilde{g}(x)$, $gcd(f_0(x)\bar{f}_1(x)-f_1(x)\bar{f}_0(x),\frac{x^n-1}{g(x)})=1$, so $([g(x)f_0(x)],[g(x)f_1(x)])$ can generate a 1-generator quasi-cyclic symplectic LCD code. Then, using Magma \cite{bosma1997magma} we can calculate this code have parameters $[42,18,9]_2^s$, whose symplectic weight distribution is  
		$\mathbf{w}_s(z)=1+448z^{9}+1344z^{10}+3906z^{11}+9051z^{12}+18753z^{13}+\dots+609z^{21}$.
		Therefore, a trace Hermitian ACD code with parameters $(21,9,9)_4$ exists. It should be noted that the best known Hermitian LCD code in \cite{ishizuka2022construction} with length 21 and dimension 9 has parameters $[21,9,8]_4$, so our symplectic construction has better performance.
	\end{example}

	\section{Conclusion}\label{sec6}
	
	In this work, we propose a new characterization for LCD codes, which allows us to determine the complementary duality of linear codes from the codeword level. Furthermore, depending on this result, we determine the sufficient and  necessary conditions for one-generator quasi-cyclic codes to be LCD codes concerning Euclidean, Hermitian, and symplectic inner products. Finally, we construct many Euclidean, Hermitian, and symplectic quasi-cyclic LCD codes to show that quasi-cyclic codes can be utilized to construct good LCD codes. 
	
	In the future, one possible extension would be to consider the sufficient and necessary condition for $h$-generator quasi-cyclic codes to be LCD. On the other hand, it will be an interesting and challenging problem to construct trace Hermitian ACD codes superior to optimal Hermitian LCD codes.

	\section*{Conflict of Interest}
	
	The authors have no conflicts of interest to declare that are relevant to the content of this article.
	
	\section*{Data Deposition Information}
	Our data can be obtained from the authors upon reasonable request.
	
	
\bibliographystyle{ieeetr}
\bibliography{reference}

	\section*{Appendix} 
	
	All the quasi-cyclic codes in this appendix with generators $([g(x)],[g(x)f(x)])$ or $([g(x)]$, $[g(x)f_1(x),[g(x)f_2(x)])$.
	Using Magma notation \texttt{QuasiCyclicCode(2*n, [g(x)],[g(x)*f(x)])} or 
	\texttt{QuasiCyclicCode(3*n, [g(x)],[g(x)*$f_1$(x)],[g(x)*$f_2$(x)])} , one can directly obtain the corresponding  LCD codes.

	\subsection*{A: Generators of ternary quasi-cyclic Euclidean LCD codes in Example \ref{Compar_ternary} }
	
	Ternary quasi-cyclic Euclidean LCD codes in Example \ref{Compar_ternary} with generator $([g(x)],[g(x)f(x)])$. Let $\{0,1,2\}$ denote elements in $\mathbb{F}_3$.  Set $n=13$, $q=3$ and index $\ell=2$. details are as follows.
	\begin{enumerate}
		\item $[26,4,15]_3$: $g(x)=2x^9 + 2x^8 + x^7 + x^6 + 2x^5 + x^4 + 2x^3 + 2x^2 + x + 1; f(x)=2x^3 + x^2 + 2x + 1;$
		\item 	$[26,6,13]_3$ and its dual $[26,20,4]_3$: $g(x)=2x^7 + 2x^5 + 2x^4 + x^3 + x^2 + 1;
		f(x)=x^5 + 2x^4 + x^2 + 2x + 1;$	
		\item 	$[26,7,13]_3$: $g(x)=2x^6 + 2x^5 + x^4 + x^2 + 2x + 2; f(x)=x^6 + 2x^5 + 2x^4 + 2x^3 + 2x;$
		\item 	$[26,8,11]_3$: $g(x)=x^5 + x^4 + x^3 + 2x^2 + 2x + 2; f(x)=x^7 + x^6 + x^5 + 2x^4 + x^3 + 1;$								
		\item 	$[26,9,10]_3$: $g(x)=x^4 + 2x^3 + 2x + 1; f(x)=2x^8 + 2x^7 + 2x^6 + 2x^5 + 2x^4 + x^2 + x + 2;$
		\item 	$[26,12,9]_3$  and its dual $[26,14,7]_3$: $g(x)=x + 2; f(x)=2x^{11} + 2x^{10} + x^7 + 2x^5 + 2x^3 + x^2 + 2x + 2;$		
		\item 	$[26,13,8]_3$: $g(x)=1; f(x)=x^{12} + x^9 + 2x^7 + 2x^6 + 2x^5 + 2x^4 + 2x^3 + x^2 + 2x;$				
		\item 	$[26,11,7]_3$ and its dual $[26,15,6]_3$: $g(x)=x^2 + x + 1; f(x)=2x^{10} + 2x^9 + x^7 + 2x^6 + x^5 + 2x^2 + x + 1;$						
	\end{enumerate}
	
	\subsection*{B: Generators of binary quasi-cyclic Euclidean LCD codes in Table \ref{New_B_LCD} }

\begin{enumerate}
	\item $[39,13,12]_2$: $n=13$, $\ell=3$, $g(x)=1$, $f_1(x)=x^{12} + x^7 + x^3 + x + 1$,  $f_2(x)=x^{12} + x^{11} + x^9 + x^8 + x^5 + x^3 + x^2$. 
	\item 	$[45,12,16]_2$: $n=15$, $\ell=3$, $g(x)=x^3 + 1,
	f_1(x)=x^{14} + x^{11} + x^{10} + x^8 + x^7 + x^6 + x^4 + x, f_2(x)=x^{13} + x^{11} + x^9 + x^3 + x^2 + 1;$. 	
	
	\item 	$[45,13,15]_2$: $n=15$, $\ell=3$, $g(x)=x^2 + x + 1,
	f_1(x)=x^{13} + x^4 + x^3 + x,
	f_2(x)=x^{14} + x^{12} + x^{11} + x^9 + x^5 + x^4 + x^3 + x;$
	
	\item 	$[44,22,9]_2$: $n=22$, $\ell=2$, 
	$g(x)=1, f(x)=x^{20} + x^{18} + x^{17} + x^{16} + x^{14} + x^{12} + x^{10} + x^9 + x^8 + x^6 + x^4 + x^3;$	
	
	\item 	$[46,23,10]_2$: $n=23$, $\ell=2$, 
$g(x)=1, f(x)=x^{18} + x^{16} + x^{13} + x^{12} + x^{11} + x^{10} + x^9 + x^6 + x^5 + x;$

	\item 	$[46,22,10]_2$ and its dual $[46,24,9]_2$: $n=23$, $\ell=2$, 
	$g(x)=x + 1,
	f(x)=x^{22} + x^{19} + x^{17} + x^{15} + x^{14} + x^{13} + x^{11} + x^6 + x^4 + x^3 + x + 1;$

	\item 	$[51,8,22]_2$: $n=17$, $\ell=3$, 
	$g(x)=x^9 + x^6 + x^5 + x^4 + x^3 + 1,
	f_1(x)=x^{16} + x^{15} + x^{14} + x^{12} + x^{11} + x^8 + x^2,
	f_2(x)=x^{16} + x^{15} + x^{13} + x^{11} + x^9;$		
								
	\item 	$[51,16,16]_2$: $n=17$, $\ell=3$, 
		$g(x)=x + 1,
		f_1(x)=x^{15} + x^{14} + x^{12} + x^{11} + x^9 + x^8 + x^5 + x^4 + x^3 + x,
		f_2(x)=x^{15} + x^{14} + x^{13} + x^6 + x^5 + x^4 + x^3 + x^2;$

	\item 	$[51,17,15]_2$: $n=17$, $\ell=3$, 
	$g(x)=1,f_1(x)=x^{15} + x^{14} + x^{10} + x^8 + x^7 + x^6 + x^5 + x^3,f_2(x)=x^{14} + x^{12} + x^{11} + x^{10} + x^8 + x^7 + x^6 + x^5;$					
\end{enumerate}

\end{document}